\documentclass[10pt, journal,transmag]{IEEEtran}
\usepackage{amsfonts}

\usepackage{blindtext}
\usepackage[table]{xcolor}
\usepackage{qtree}
\usepackage{multirow}
\usepackage{comment}
\usepackage{mathrsfs}
\usepackage{subcaption}

\usepackage{cite}
\usepackage{amsmath,amssymb,amsfonts,amsthm}
\usepackage{algorithm}
\usepackage[noend]{algpseudocode}
\usepackage{graphicx}
\usepackage{textcomp}
\def\BibTeX{{\rm B\kern-.05em{\sc i\kern-.025em b}\kern-.08em
    T\kern-.1667em\lower.7ex\hbox{E}\kern-.125emX}}

\newcolumntype{C}[1]
{>{\centering\let\newline\\\arraybackslash\hspace{0pt}}m{#1}}
\usepackage{setspace,tikz}

\usepackage{pgfplots}
\pgfplotsset{compat=newest}
\usetikzlibrary{patterns}

\usepackage{diagbox}

\newtheorem{theorem}{Theorem}
\newtheorem{corollary}{Corollary}

\newtheorem{example}{Example}

\DeclareMathOperator*{\argmax}{\arg\!\max}

\begin{document}

\bstctlcite{IEEEexample:BSTcontrol}



\title{Packet Loss Recovery in Broadcast for Real-Time Applications in Dense Wireless Networks}


\author{\IEEEauthorblockN{Afshin~Arefi\IEEEauthorrefmark{1}, 
and Majid~Khabbazian\IEEEauthorrefmark{2},~\IEEEmembership{Senior Member,~IEEE},\\}
\thanks{\IEEEauthorblockA{\IEEEauthorrefmark{1}arefi@ualberta.ca},
\IEEEauthorblockA{\IEEEauthorrefmark{2}mkhabbazian@ualberta.ca},
}}

\maketitle

\begin{abstract}
  Packet loss recovery in wireless broadcast is challenging,
  particularly for real-time applications which have strict and short delivery deadline.
  To recover the maximum number of lost packets within a short time, 
  existing packet recovery solutions often rely on instantly decodable network coding (IDNC).
  Some of these solutions can recover nearly the maximum number of lost packets possible
  at the cost of collecting feedback from all (or a large percentage of) users.
  This is impractical in dense networks.
  In addition, their runtime grows with the number of users, which is not desirable due to the urgent delivery deadline of real-time applications. 
  
  In this work, we introduce RIDNC, a random encoding approach to IDNC.
  We propose RACE, a light RIDNC encoder that can recover nearly as many lost packets as the optimal RIDNC encoder.
  We compare RACE with the CrowdWiFi encoder, a high performing packet loss recovery solution used in CrowdWiFi,
  a commercial system for broadcasting live video in dense networks.
  We show that RACE is up to two orders of magnitude faster than  the CrowdWiFi encoder, 
  and recovers more lost packets in practice, where there is not enough time to collect feedback
  from many users.

%

\end{abstract}


\section{Introduction}

With the increase in popularity of mobile devices, the need for greater quality
and higher bandwidth in wireless networks has increased exponentially.
Cisco Visual Networking Index report shows that the overall mobile traffic will
rise from 11 ExaBytes per month in 2017 to more than 48 ExaBytes per month in
2021 (doubling every two years) \cite{cisco2017paper}.
At the same time video content is rising as well, and by 2021, it
will cover 82\% of all Internet traffic; 16\% of this video content will be
live video streams \cite{cisco2017paper}.

In this paper, we consider wireless broadcast of live media in a network with a single transmitter and many receivers within its transmission range.
This has many applications.
For example, a wireless transmitter can broadcast the bird-view of a parking lot to help vehicles find an open parking spot.
In large sport stadiums, such transmitter can enable spectators watch the replay on their phones, 
or in concerts, it allows far away seats to get a better view of the stage on their tablets.
Another example is to provide students in large theatre classes with a better view of
the board on their laptops.
In all these scenarios, we need to broadcast live video to hundreds or even thousands of users/receivers 
within the transmitter's range.

\textbf{Handling packet losses.}
In wireless communication, packet loss is common because of various channel
impairments such as mutli-path fading.
A basic method to recover lost packets is to simply retransmit each lost packet.
For example, in Wi-Fi (IEEE 802.11), the transmitter retransmits a packet if it
does not receive an acknowledgment  from the receiver.
This is an effective way to recover lost packets when there is only a single
receiver.
When there are multiple receivers (i.e., in case of broadcast), however, there are
more effective solutions. 

As a simple example, consider a transmitter with three receivers within its transmission range.
Suppose that after transmitting three packets, the transmitter realizes (through feedback/acknowledgments) that
receiver one is missing packet one, receiver two is missing packet two and
receiver three is missing packet three. 
To recover packets, the transmitter can transmit all three packets one more
time.
However, a more efficient solution is to XOR the three packets to construct a
coded packet, and transmit the coded packet only.
This way, each receiver can recover its missing packet by XORing the
coded packet with the two packets it possesses.

\textbf{Real-time applications.}
Real time applications can tolerate some packet losses, 
but have urgent packet delivery deadlines,
which means that the transmitter has limited time to recover a lost packet.
Our objective is, therefore, not to recover all lost packets as the transmitter may not have enough time for it.
Instead, similar to~\cite{ferreira2014real} and~\cite{le2017recovery}, our aim is to 
recover as many lost packets as possible within a short time interval, that is with limited number of transmissions.
Random Network Codes (RNC) and fountain codes are not suitable for this aim, as they require receivers 
to receive several coded packets before they can start decoding any packet.
Therefore, instead of using these codes, we focus on codes that allow instantaneous decoding at the receivers.

\textbf{Instantly decodable network codes (IDNC).} IDNC is an attractive solution for packet loss recovery in 
broadcast of real-time applications \cite{keller2008online}.
However, existing IDNC encoders typically require collecting feedback from all or most of the receivers.
This is a significant overhead in dense networks.
To get a numeric intuition on the amount of this overhead, 
consider a network with $100$ receivers/users.
Suppose that feedback is transmitted using UDP over IP in a 802.11 network.
The header size of UDP, IP and 802.11 are respectively $8$ bytes, $20$ bytes,
and $34$ bytes.
This leads to at least $62$ bytes of header overhead per user.
Assuming that the size of a data packet is $1000$ bytes, the bandwidth
requirement for collecting feedback from $100$ users is about that required to
transmit six full data packets.
In practice, the overhead of collecting feedback may be even higher as nodes have to
employ large back-offs to reduce the number of collisions.

The second challenge with IDNC is that finding an optimal code is NP-hard.
There are suboptimal IDNC encoders in the literature.
However, the computational complexity of these encoders grows with the number of users.
When the number of users grows, these solutions start to
become slow and unsuitable for real time applications.
Some works in the literature (e.g., \cite{bejerano2016amuse}) tried to tackle this problem by assigning
cluster heads and collecting feedback from them (instead of all users).
These algorithms are relatively complicated and rely on the spacial packet loss correlation.
Some existing studies show that such correlation is low~\cite{SalyersSP08}.

\textbf{Contributions.} 
    In this work, we introduce RIDNC, a random approach to IDNC encoding.
    We prove that, in dense networks, a single RIDNC packet 
    can recover nearly as many lost packets as possible.
    Using the RIDNC approach, we propose RACE, a fast \underline{RA}ndom IDN\underline{C} \underline{E}ncoder 
    that works with limited feedback from users,
    yet recovers as many lost packets as the optimal RIDNC encoder.    
    Using simulations, we compare RACE with one of the best IDNC-based packet recovery solutions, 
    verify the superior performance of RACE in speed and in recovering lost packets, 
    and confirm its low communication overhead in dense networks.

    The rest of the paper is organized as follows.
    Section~\ref{sec:Rel} covers related work in the literature.
    In Section~\ref{sec:PD}, the system model and problem definition are presented.
    In Section~\ref{sec:Main}, we motivate the use of RIDNC in dense networks, and propose the RACE encoder.
    Simulation results are presented in Section~\ref{sec:sim}, and the paper is
    concluded in Section~\ref{sec:Con}.

\section{Related Work}
\label{sec:Rel}

Ahlswede et al. first introduced network coding in \cite{ahlswede2000network}
and showed that it can improve throughput.
Since then, there has been extensive research work on designing coding-based
solutions for different applications. 
In the case of single-hop wireless broadcast, the transmitter can use random
linear network codes (RLNC) \cite{eryilmaz2008delay,yu2014instantly,
ho2006random,lucani2009broadcasting,lima2007random,li2003linear,li2005theory,
lucani2009random} and Raptor codes \cite{shokrollahi2006raptor} to deliver
packets with minimum number of transmissions, and with low coding complexity,
respectively.
These solutions are, however, not suitable for real-time applications such as
live video streaming, as packets have strict delivery deadline in such
applications.
To meet the delivery deadline, a received coded packet needs to be decoded
within a short time window; otherwise it would be useless hence discarded.
Instantly Decodable Network Codes (IDNC) \cite{nguyen2009network,
costa2008informed,rayanchu2008loss} are a family of Opportunistic Network Codes
(ONC) \cite{seferoglu2007opportunistic,chen2007opportunistic,
seferoglu2009video} that minimize this decoding time at the cost of lower
throughput than RLNCs. 

IDNCs have the following distinguishing properties:
1) a coded packet is constructed by simply XORing a number of plain packet.
2) a received coded packet is either instantly decoded using the past decoded
packets or discarded (i.e., it is not stored for later decoding).
Because of the latter property, IDNCs have been the subject of several work
studying real time multimedia broadcast \cite{fragouli2007feedback,
eryilmaz2008delay,yu2014instantly}.

Eryilmaz et al. \cite{eryilmaz2008delay} study the delay gain of coding in unreliable networks.
They find closed-form expressions for the delay performance with or without coding, and show significant delay gains when coding is used.
They further extend their results to general network topologies. 
%
%
%
%
%
Yu et al. \cite{yu2014instantly} analyze the tradeoff between throughput and
decoding delay, and examine the performance gap between RLNC and IDNC.
They also propose a number of coding solutions with varying delay-throughput tradeoffs.
%
Fragouli et al. \cite{fragouli2007feedback} examine different usages of
feedback in networks with coding capabilities, and illustrate benefits including
adaptive parameter optimization to provide better quality of services.
They also consider the possibility of using network coding to the feedback packets,
and examine design of acknowledgment packets.

The most relevant work to ours among these studies is \cite{le2013instantly,
le2017recovery}, in which the authors show that the problem of finding a code
that is instantly decodable by the maximum number of receivers is NP-hard.
They also propose a polynomial time code construction method for the case where
all receivers experience an identical erasure rate.

There is a large body of work on IDNCs that aim to reduce the completion time,
which is the time needed to deliver packets to all the receivers. 
Some of these works assume that some packets are more important to be delivered
than other packets \cite{seferoglu2009video,muhammad2013instantly}. 
As argued in~\cite{le2017recovery} minimizing the completion time is not the
right objective for real-time applications.
What is important in such applications is rather to deliver as many packets as
possible within a strict delivery deadline.
For this reason, similar to \cite{katti2008xors,le2013instantly,le2017recovery,
arefi2018blind}, we focus on designing IDNCs that are instantly
decodable by as many users as possible. 

There are two other problem setups in the literature that resemble the problem
considered in this paper.
These problems are index coding \cite{birk2006coding} and data exchange
\cite{el2010coding}. 
The objective of both index coding and data exchange is to minimize the number
of transmissions to achieve a certain goal.
For example, in the index coding problem, each receiver demands a single packet
from the set of packets at the transmitter.
The objective is to satisfy all these demands with minimum number of
transmissions.
The objectives considered in index coding and data exchange are different from our objective, which is to
construct a single code that is instantly decodable by the maximum number of
users.

\section{Problem Statement and System Model}
\label{sec:PD}

Similar to the work of Le et al. \cite{le2013instantly,le2017recovery}, we
define the problem as follows.
We consider a single wireless base station (also referred to as 
transmitter) and $N$ users (also referred to as receivers)
$\mathcal{U}=\{u_1,u_2,u_3,\cdots,u_N\}$.
We assume that all users are within the transmission range of the base station.

The base station has $M$ packets $\mathcal{P}=\{p_1,p_2,p_3,\cdots,p_M\}$ to
broadcast to all $N$ users.
It first broadcasts these $M$ packets using $M$ transmissions.
This step is called the \emph{initial transmission phase}.
The base station is then granted a short time to recover as many
lost packet as possible by transmitting few coded packets.

In IDNC (and RIDNC), the encoder at the transmitter constructs a coded packet by XORing a subset of packets $\mathcal{C}\subseteq \mathcal{P}$.
The number of packets in $\mathcal{C}$ (i.e., $|\mathcal{C}|$) is referred to as the \emph{code degree}.
At the receiver side, a received coded packet is used by the decoder to recover/decode a lost packet; 
if not successful, the coded packet is discarded (i.e., it is not buffered for later use).
The decoder is able to recover a lost packet iff the receiver has exactly one lost packet in the set $\mathcal{C}$.

\textbf{Objective.} 
As in \cite{le2013instantly,le2017recovery}, our objective is to recover as many lost packets as possible 
using a few number of coded packet transmissions. 
This objective is motivated by the following facts:
\begin{enumerate}  
  \item Real-time applications often have a strict and urgent packet delivery deadline. 
    Therefore, the transmitter has limited time, hence limited number of transmission opportunities to recover lost packets.
   \item Real-time applications can often tolerate some packet losses.
\end{enumerate}

\textbf{Channel model.}
  To analyze the performance of encoders, we model the channel between the transmitter and users by 
  a packet erasure channel, where a packet is either correctly received or lost.
  The packet lost probability of user $u_i$, $1\leq i\leq N$, is denoted by $\mathcal{E}_i$ and is referred to as the packet erasure rate of $u_i$.
  We assume that different users can have different packet erasure rates, and that packet receptions at different users are independent. 

\section{Random Instantly Decodable Network Coding}
\label{sec:Main}

  Random Instantly Decodable Network Coding (RIDNC) is a random encoding approach to IDNC.
  In RIDNC, the encoder first decides on a \emph{code degree} $d$.
  It then simply selects $d$ packets uniformly at random, and XORs them to generate a coded packet.
  The simplicity of RIDNC allows design of fast encoders.
  Fast encoders are attractive for packet loss recovery in broadcast of real-time applications
  because these applications often have urgent packet delivery deadlines.


%
%
  

\subsection{Asymptotic Optimality of RIDNC}
\label{sec:optimal}
  The main task of a RIDNC encoder is to find a good code degree.
  Different code degrees result in different packet recovery performances.  
  For example, suppose that every user is missing exactly one packet out of $M$ packets.
  A small code degree, in this case, would result in a poor packet recovery, while the code degree of $M$
  would result in recovering all the lost packets (assuming that the coded packet is received by all the users).

  Suppose the transmitter is granted a transmission (after the initial transmission phase) to recover as 
  many lost packets as possible.
  The next theorem states that, in a dense network, if a RIDNC encoder chooses the right code degree, it can recover nearly
  as many lost packet as possible by any other packet recovery solution.
  Recall that $N$ and $M$ denote the number of users and the number of packets, respectively,

   \begin{theorem}
  \label{thm:stat}
    Let $G_{R}$ denote the expected number of recovered lost packets when the optimal code degree is used by the RIDNC encoder.
    Let $G_{opt}$ denote the maximum number of lost packets that is possible to recover by any solution. \\
    Then, for any arbitrary small positive real numbers $\delta$ and $\epsilon$, we have
    \[
      G_{R} \geq (1-\delta)\cdot G_{opt},
    \]
    with probability at least $1-\epsilon$, if
    \[
      N\geq \left(\frac{3\ln 2\ln\frac{1}{\epsilon}}{p^*\cdot\delta^2}\right)\cdot M
    \]  
  \end{theorem} 
  \begin{IEEEproof}  
    See Appendix~\ref{app:stat}.
  \end{IEEEproof}

  
  \subsection{RACE: The Proposed RIDNC encoder}
%
    The core of RACE (and any RIDNC encoder) is the code degree calculation.
    When the code degree is calculated, the task of an RIDNC encoder is simple: it XORs a random subset of packets.
    The challenge of RIDNC encoder is to estimate the optimal code degree as accurate and as fast as possible 
    with preferably limited feedback from users.
    Following, we first explain how RACE, our proposed RIDNC encoder, 
    estimates the optimal code degree for the first coded packet (Algorithm~\ref{alg:general}). 
    The general case of code degree estimation is then presented in Algorithm~\ref{alg:extended}.
    
   
    \textbf{Estimating the optimal degree for the first coded packet.}
    Let the random variable $Y_d$ be the number of lost packets recovered when the transmitted coded packet
    is the XOR of $d$ packets selected uniformly at random from $\mathcal{P}$. 
    By definition, the optimal code degree $d^*$ is  
    \[
      d^*=\argmax_{1\leq d \leq M} E[Y_d].
    \]
    For the first coded packet, we have
    \[
     E[Y_d]=\sum_{i=1}^{N}d\cdot \mathcal{E}_i(1-\mathcal{E}_i)^{d-1}
    \]
    where $\mathcal{E}_i$ is the erasure rate of user $u_i$.
    Therefore, if the encoder knows the erasure rates $\mathcal{E}_i$, it can calculate the optimal code degree $d^*$
    for the first coded packet. 
    RACE, however, has no prior information on the erasure rates of users.
    To estimate the optimal code degree, it therefore collects limited feedback from a subset of users right after the initial transmission phase.


    Let $S$ be the index of users from which feedback is collected.
    The feedback collected from user $u_i$, $i\in S$, is the number of packets $u_i$ has received during the initial transmission phase.
    After collecting all these feedback, RACE estimates the optimal code degree $d^*$ as
    \begin{equation}
    \label{equ:estimate}
      \hat{d}=\argmax_{1\leq d \leq M}\left( \sum_{i\in S}  \frac{{h_i \choose d-1}(M-h_i) } {{M \choose d}}\cdot \frac{h_i+1}{M+2} \right).
    \end{equation}

    The estimation (\ref{equ:estimate}) is inspired by Theorem~\ref{thm:expected}.
    In the theorem, the assumption that erasure rates come from a uniform 
    distribution---while they likely  come from non-uniform distributions in practice---is only to estimate the optimal code degree. 
    One may justify this assumption by the fact that the transmitter has no information about erasure rates prior to collecting feedback. 
    Note that any prior information about the erasure rates (e.g, knowing users' erasure rate distributions) can be used in Theorem~\ref{thm:expected} to improve the estimate of the optimal code degree.
    Nevertheless, simulation results show that this estimate is virtually as good as the optimal code degree.

   An advantage of (\ref{equ:estimate}) is that
   a major part of it can be pre-computed: as stated in Corollary~\ref{cor:type}, evaluating $\hat{d}$ 
   would require only a single matrix multiplication of complexity $\mathcal{O}(M^2)$.
   Note that this computational complexity does not grow with the number of users $N$.

    \begin{theorem}
    \label{thm:expected}
      Suppose that the erasure rates of users are uniformly distributed in the interval [0,1].       
      Let $H_i$ be the event that user $u_i$ has received $h_i$ packets in the initial transmission phase.   
      Let the random variable $Y'_d$ be the number of recovered packets when the transmitted 
      coded packet is the XOR of $d$ packets selected uniformly at random. \\ 
      Then
      \[
        E[Y'_d|H_1, H_2, \ldots, H_N]= \sum_{i=1}^{N}\left( \frac{{h_i \choose d-1}(M-h_i) } {{M \choose d}}\cdot \frac{h_i+1}{M+2} \right).
      \]
    \end{theorem}
  \begin{IEEEproof}  
    See Appendix~\ref{app:expected}.
  \end{IEEEproof}

    \begin{corollary}
    \label{cor:type}
    Let $n_j$ be the number of users that have received $j$, $0\leq j \leq M$,
    packets during the initial transmission phase, and $\mathbf{V}=[n_0, n_1, \ldots, n_{M-1}]^T$.
    Let $\mathbf{\Pi}$ be a $M \times M$ matrix, such that
    \[
       \mathbf{\Pi}_{i,j}= \frac{{j \choose i-1}(M-j) } {{M \choose i}}\cdot \frac{j+1}{M+2},
    \]
    where $\mathbf{\Pi}_{i,j}$ denotes the element in the $i$th row and $j$th column of $\mathbf{\Pi}$.
    Then
      \[
        \hat{d}=   \argmax_{1\leq d\leq M}\left(\mathbf{\Pi} \times \mathbf{V}\right)_d
      \]
    where $\left(\mathbf{\Pi} \times \mathbf{V}\right)_i$ denotes the $i$the element of $\mathbf{\Pi} \times \mathbf{V}$.
    
    \end{corollary}
    \begin{proof}
      We have
      \[
      \begin{split}
         \hat{d}
         &=\argmax_{1\leq d \leq M}\left( \sum_{i\in S}  \frac{{h_i \choose d-1}(M-h_i) } {{M \choose d}}\cdot \frac{h_i+1}{M+2} \right)\\
        &=\argmax_{1\leq d \leq M}\left(\sum_{j=0}^{M} \left( \frac{{j \choose d-1}(M-j) } {{M \choose d}}\cdot \frac{j+1}{M+2} \right)\cdot n_j\right)\\
        &=\argmax_{1\leq d \leq M}\left( \sum_{j=0}^{M} \Pi_{d, j}\cdot n_j \right)\\
        &= \argmax_{1\leq d \leq M} \left(\mathbf{\Pi} \times \mathbf{V}\right)_{d}. \\
      \end{split}
      \]
    \end{proof}

\noindent
Algorithm~\ref{alg:general} shows how RACE estimates the optimal degree for the first coded packet.
As shown in the algorithm, the matrix $\mathbf{\Pi}$ is pre-computed, and the matrix $\mathbf{V}$ 
is simply filled in with the collected feedback.
Therefore, the main computation of Algorithm~\ref{alg:general} is to calculate the product $\mathbf{\Pi}\times \mathbf{V}$
which can be done in $\mathcal{O}(M^2)$.

\begin{algorithm}
\caption{Code degree estimation for the first coded packet}
\label{alg:general}

\begin{algorithmic}[1]
\State $\textit{// \textbf{Pre-computation}}$
\State $M \gets |\mathcal{P}|$ \textit{// The umber of packets}
\State $\mathbf{\Pi}\gets \textit{Zero matrix of size $M \times M$ }$
\ForAll{$i \in [1,\ldots,M]$} \textit{//  Number of packets to be XORed}
  \ForAll{$j \in [0,\ldots,M-1]$} 
    \State $\mathbf{\Pi}_{i,j} \gets \left( \frac{{j \choose i-1}(M-j) } {{M \choose i}}\cdot \frac{j+1}{M+2} \right)$
  \EndFor
\EndFor
\State $\textit{// \textbf{Realtime computation}}$
\State $n \gets \textit{Number of feedback samples}$
\State $S \gets collectRandomSamples(n)$
\State $\mathbf{V}\gets \textit{Zero vector of size $M$ }$
\ForAll{$u_i, i \in S$}
  \State $j \gets $ \# packets received by $u_i$ \textit{// Ignore if  $u_i$ has received all}
  \State $\mathbf{V_j} \gets \mathbf{V_j} + 1$ \textit{// Ignore if  $u_i$ has received all}
\EndFor
\State $Gain\gets 0$
\State $Num\gets 0$
\State $\mathbf{G}\gets \mathbf{\Pi} \times \mathbf{V}$
\ForAll{$d \in [1,\ldots,M]$} \textit{// code degree}
  \If{$\mathbf{G}_i > Gain$}
    \State $Gain\gets \mathbf{G}_d$
    \State $\hat{d}\gets d$
  \EndIf
\EndFor
\State \Return $\hat{d}$ 
\end{algorithmic}
\end{algorithm}

\pagebreak
\begin{example}

Suppose the number of packets is $M=10$.
The matrix $\mathbf{\Pi}$ is pre-computed as
\[
\scriptsize
\mathbf{\Pi}=
\begin{bmatrix}

 \mathbf{.08} & \mathbf{.15} & \mathbf{.20} & \mathbf{.23} &
 \mathbf{.25} & \mathbf{.25} & \mathbf{.23} & \mathbf{.20} &
 \mathbf{.15} & \mathbf{.08}  \\
 
 0 & \mathbf{.03} & \mathbf{.09} & \mathbf{.16} &
 \mathbf{.22} & \mathbf{.28} & \mathbf{.31} & \mathbf{.31} &
 \mathbf{.27} & \mathbf{.17} \\
 
 0 & 0 & \mathbf{.02} & \mathbf{.06} &
 \mathbf{.12} & \mathbf{.21} & \mathbf{.29} & \mathbf{.35} &
 \mathbf{.35} & \mathbf{.25} \\
 
 0 & 0 & 0 & \mathbf{.01} &
 \mathbf{.05} & \mathbf{.12} & \mathbf{.22} & \mathbf{.33} &
 \mathbf{.40} & \mathbf{.33}  \\
 
 0 & 0 & 0 & 0 &
 \mathbf{.01} & \mathbf{.05} & \mathbf{.14} & \mathbf{.28} &
 \mathbf{.42} & \mathbf{.42}  \\
 
 0 & 0 & 0 & 0 &
 0 & \mathbf{.01} & \mathbf{.07} & \mathbf{.20} &
 \mathbf{.40} & \mathbf{.50}  \\
 
 0 & 0 & 0 & 0 &
 0 & 0 & \mathbf{.02} & \mathbf{.12} &
 \mathbf{.35} & \mathbf{.58}  \\
 
 0 & 0 & 0 & 0 &
 0 & 0 & 0 & \mathbf{.04} &
 \mathbf{.27} & \mathbf{.67}  \\
 
 0 & 0 & 0 & 0 &
 0 & 0 & 0 & 0 &
 \mathbf{.15} & \mathbf{.75}  \\
 
 0 & 0 & 0 & 0 &
 0 & 0 & 0 & 0 &
 0 & \mathbf{.83}  \\
\end{bmatrix}
\]
Suppose the following vector is the summary of the collected feedback.
\[
\mathbf{V}=
\begin{bmatrix}
0 & 0 & 0 & 1 & 1 & 2 & 4 & 4 & 7 & 9 \\
\end{bmatrix}^T,
\]
For instance, the value of $\mathbf{V}$ at index eight is seven. 
This implies that seven users (out of all users from which feedback was collected) have received 8 packets.
Given $\mathbf{\Pi}$ and $\mathbf{V}$, the product $\mathbf{\Pi}\times \mathbf{V}$ is computed as

\[
\scriptsize
  \begin{bmatrix}
  0.00 & 4.51 & 6.78 & 7.86 & 8.31 & \underbrace{\mathbf{8.44}}_{\text{index } i=5} & 8.39 & \ldots & 7.50 \\
  \end{bmatrix}^T.
\]
The largest element of the above vector is at index $i=5$, thus Algorithm~\ref{alg:general} returns $\hat{d}=5$ as the code degree.

\end{example}

%
%

\textbf{Estimating optimal code degrees for multiple packets.}
    After transmission of each coded packet, the optimal code degree can change.
    To estimate the new optimal code degree, RACE does not request feedback from users  
    as this increases the communications overhead and delay.
    Instead, as shown in Algorithm~\ref{alg:extended}, 
    it estimates the number of packets each user has using its latest information about the user, 
    and the code degree of the last transmitted coded packet.

    To this end, Algorithm~\ref{alg:extended} maintains a $M\times M$ matrix  $\mathbf{D}$, a generalization of the vector $\mathbf{V}$ 
    in Algorithm~\ref{alg:general}. 
    The value of $\mathbf{D}_{i,j}$ is an estimate of the number of users that have received $i$ packets in the initial transmission phase, 
    but currently have $j$ ($j\geq i$) packets because of possible recoveries by the coded packet transmissions.
    Initially, $\mathbf{D}$ is filled with the collected feedback from users, so the initial values of $\mathbf{D}$ are all accurate.
    After each transmission, the values of $\mathbf{D}$ are updated as 
    \[
      \mathbf{D}_{i,j} \leftarrow 
      \mathbf{D}_{i,j} -\alpha_{c,j}\cdot\mathbf{D}_{i,j}
      + \alpha_{c, j-1}\cdot\mathbf{D}_{i,j-1},
    \]
    where 
    \[
      \alpha_{c,j}= \frac{{j \choose c-1}(M-j) } {{M \choose c}}\cdot \frac{i+1}{M+2},
    \]
    and $c$ is the code degree of the last transmitted coded packet.
    By Theorem~\ref{thm:expected}, the parameter $\alpha_{c,j}$ is an estimate of the probability that a node 
    which has received $i$ packets in the initial transmission phase,
    and currently has $j$ packets has benefited from the last transmitted coded packet.
    In other words, $\mathbf{D}_{i,j}$ is an estimate of the number of users that  received $i$ packets in the initial transmission phase
    but currently have $j\geq i$ packets.
    Using the updated matrix $\mathbf{D}$, and by simply applying Theorem~\ref{thm:expected}, 
    Algorithm~\ref{alg:extended} esimates the code degree for the next coded packet.

\begin{algorithm}
\caption{Code degree estimation in RACE}
\label{alg:extended}

\begin{algorithmic}[1]
\State $M \gets |\mathcal{P}|$ \textit{// number of packets}
\State $n \gets \textit{Number of feedback samples}$
\State $r \gets \textit{Number of coded packets}$
\State $R \gets \emptyset$ \textit{// Set of constructed coded packets}
\State $S \gets collectRandomSamples(n)$
\State $\mathbf{D}\gets \textit{Zero matrix of size $(M+1)\times(M+1)$}$
\ForAll{$u_i, i \in S$}
  \State $i \gets $\# packets received by $u_i$ 
  \State $\mathbf{D_{i,i}} \gets \mathbf{D_{i,i}} + 1$
\EndFor
\ForAll{$i \in [1,\ldots,r]$}
  \State $Gain\gets 0$
  \State $Num\gets 0$
  \ForAll{$d \in [1,\ldots,M]$} 
    \State $\mathbf{D'}\gets\mathbf{D}$
    \State $G \gets 0$
    \ForAll{$j \in [0,\ldots,M-1]$} 
      \ForAll{$k \in [j,\ldots,M-1]$} 
        \State $tmpGain \gets \left( \frac{{k \choose d-1}(M-k) } {{M \choose d}}\cdot \frac{j+1}{M+2}\cdot\mathbf{D}_{j,k} \right)$
        \State $G \gets G + tmpGain$
        \State $\mathbf{D'}_{j,k+1} \gets \mathbf{D'}_{j,k+1} + tmpGain$
        \State $\mathbf{D'}_{j,k} \gets \mathbf{D'}_{j,k} - tmpGain$
      \EndFor
    \EndFor
    \If{$G > Gain$}
      \State $Gain\gets G$
      \State $Num\gets d$
      \State $\mathbf{D''}\gets\mathbf{D'}$
    \EndIf
  \EndFor
  \State $\mathbf{D}\gets\mathbf{D''}$
  \State $\hat{d}_i \gets Num$ \textit{// estimated optimal code degree of the $i$th packet}
\EndFor
\State \Return $(\hat{d}_1, \hat{d}_2, \ldots, \hat{d}_r)$
\end{algorithmic}
\end{algorithm}

\textbf{Collecting feedback.} 
To statistically represent the set of receivers/users in the network, the receivers that send feedback are selected uniformly at random.
To collect feedback from the selected users, different mechanisms can be employed.
Following, we explain two possible feedback mechanisms. 

  For the first mechanism, we assume that the transmitter is aware of the IDs of the receivers. 
  This is a reasonable assumption because receivers typically connect to the transmitter to receive the multicast service.
  To request feedback, the transmitter can then randomly select receivers, and embed the list of selected IDs in a packet to notify the selected receivers.
  To minimize collisions, the selected receivers can transmit their feedback in the same order as their IDs appear in the packet.
  Note that even without following such order of transmission, the MAC layers of the selected users are there to handle the contention.
  Also, our solutions require only a small number of feedback.
  Therefore, there is a smaller chance of collisions compare to the case where feedback is collected from all/majority of users.  
  In addition, as shown in the simulations, our solutions are not sensitive to few possible feedback losses.  
  
  For the second mechanism, we assume that the transmitter does not have the IDs of the receivers, but has an estimate of the number of receivers.
  In this case, the transmitter can ask the receivers to send feedback with a probability $p$, set by the transmitter.
  For instance, if the transmitter estimates the number of receivers to be 100, and it wishes to collect about 20 feedback,
  then it can set $p$ to be equal or slightly higher than 0.2.
  This solution has lower communication overhead than the first solution, and does not require knowledge of IDs of the receivers.


\section{Simulation Results }
\label{sec:sim}

\textbf{Comparing RACE with the Optimal RIDNC encoder.}
We first compare the performance of RACE with the optimal RIDNC encoder. 
Recall that, by Theorem~\ref{thm:stat}, in dense network, the optimal RIDNC is capable of recovering virtually as many lost packets as possible by any other recovery solution.

In the optimal RIDNC, the transmitter knows the erasure rates of all users, 
and receives feedback from all users after each transmission to calculate the optimal code degree.
In RACE, however, the optimal code degree is estimated using a one-time feedback (right after the initial transmission phase), 
and without the knowledge of erasure rates.

Figures~\ref{fig:multi_RIDNC10}, \ref{fig:multi_RIDNC20}, and \ref{fig:multi_RIDNC30} compare the gain of RACE and the optimal RIDNC encoder for various
number of coded packet transmissions. 
The gain is calculated as the total number of packets recovered divided by the total number of users, averaged over 1,000 simulation runs.
As shown, RACE performs very close to the optimal RIDNC.
This is despite the fact that RACE has no knowledge of erasure rates, and receives feedback from users only once.

\textbf{Sensitivity of RACE to the number of feedback.}
In the next set of simulations, we evaluate the sensitivity of the gain of RACE to the number of feedback.
Recall that RACE collects feedback only once (before the transmission of the first coded packet).
Figures~\ref{fig:stat} and \ref{fig:multiStat} show the result of our simulations for a dense network with 500 users.
The dashed line shows the gain of RACE when feedback is collected from all users.
The solid curve is the gain of RACE versus the number of feedback requested\footnote{Feedback packets get lost just like other packets.}.
This result shows that RACE needs only a small number of feedback to nearly achieve its full gain.
For instance, it shows that RACE can achieve $90\%$ of its full gain
using feedback from up to $\sim2\%$ of users.

\textbf{Comparing RACE with the CrowdWiFi encoder.}
Among the existing solutions, two were reported to achieve near optimal gain: the recovery algorithm used by 
CrowdWiFi~\cite{ferreira2014real} (referred to as the CrowdWiFi encoder), and the Multi-Slot Max Clique~\cite{le2017recovery}.
The CrowdWiFi encoder has a better runtime\footnote{The runtime of CrowdWiFi encoder was wrongly reported to be exponential in~\cite{le2017recovery}.} 
than Multi-Slot Max Clique, thus we used it as our baseline to evaluate the
gain and speed of RACE. 

CrowdWiFi by Streambolico~\cite{CrowdWiFi} is a commercial system for broadcasting live video in dense networks such as stadiums, and conferences.
A major contributor in the high performance of CrowdWiFi is its encoder, which works as follows.
In the first step, the CrowdWiFi encoder selects a packet, say $c$, that is wanted by the most number of users.
At step $i>1$, the encoder selects a packet $p$ such that $c\oplus p$ is instantly decodable by more users than $c$.
If such a packet $p$ exists, it replaces $c$ with $c\oplus p$ and moves on to the next step; otherwise, it stops and returns $c$ as the 
coded packet.

We implemented the CrowdWiFi encoder, and compared it against RACE.
As shown in Figures~\ref{fig:RvsCGainE20} and~\ref{fig:RvsCGainE10}, 
RACE can recover more lost packets than the CrowdWiFi encoder where the number of feedback is small.
This is an advantage because the transmitter does not have
enough time to collect many feedback when broadcasting for real-time applications.
As the number of collected feedback increases, the gap between the gain of the two encoders reduces,
and at some ``cross point'' the CrowdWiFi encoder can recover the same number of lost packets as RACE.
To get to this cross point, however, a considerable number of feedback has to be collected.
In addition, at this point, the CrowdWiFi encoder is significantly slower than RACE.
For instance, when the number of packets is $M=20$ and the maximum erasure rate is $20\%$, 
the cross point occurs around $0.15\times N$ feedback, 
as illustrated in  Figures~\ref{fig:RvsCGainE20}.
As shown in Figure~\ref{fig:multiStat10}, when M=20 and the maximum erasure rate is $20\%$, 
RACE is about a factor of $3.5F$ faster than the CrowdWiFi encoder, where
$F$ denotes the number of collected feedback.
Combining the two, we get that RACE is  
$3.5 \times (0.15 N)> 0.5 N$ faster than the the CrowdWiFi encoder  at the cross point.
This speedup factor ranges from 50 to 500 in a dense network with 100 to 1000 users.


\pagebreak
\section{Conclusion}
\label{sec:Con}
   
   In dense networks, it is not  practical to collect feedback from all or a large percentage of users.
   This is especially the case when broadcasting for real-time applications, which have urgent delivery deadlines.
   Considering this limitation, we proposed RACE, a fast and light random IDNC (RIDNC) encoder.
   We proved that, in dense networks, the optimal RIDNC encoder is capable of recovering nearly as many lost packets as possible
   by any other solution.
   Then, using simulations, we showed that RACE can recover virtually as many lost packets as the optimal RIDNC.
   Also, we compared RACE with the CrowdWiFi encoder, a fast and high performing packet recovery solution used in CrowdWiFi,
   a commercial system for broadcasting video in dense networks.
   The results show that RACE can recover more lost packets than the CrowdWiFi encoder in practice, 
   where there is not enough time to collect many feedback.
   In addition, RACE is up to two orders of magnitude faster than the CrowdWiFi encoder.   
   These make RACE a better solution than the state-of-the-art for recovering lost packets in broadcast for real-time applications 
   in dense networks.

%
%
%
%


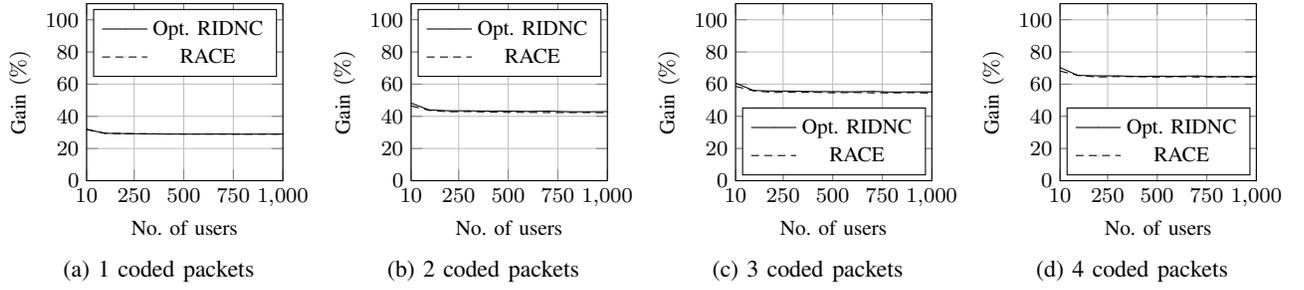
\begin{figure*}
\centering
\begin{subfigure}[b]{1.65in}
\centering
\footnotesize
\begin{tikzpicture}
\begin{axis}[ xlabel=No. of users,
              ylabel=Gain $(\%)$,
              grid=major,
              log ticks with fixed point,
              xmin=10,
              xmax=1000,
              xtick={10,250,500,750,1000},
              ymin=0,
              ymax=110,
              ytick={0,20,...,100},
              height=1.55in,
              width=1.65in,
              legend pos=north east,
              y tick label style={/pgf/number format/fixed,
              /pgf/number format/1000 sep = \thinspace}
            ]

\addlegendimage{color=black,mark=*,mark size=0}
\addlegendentry{Opt. RIDNC} 
\addlegendimage{densely dashed,color=black,mark=*,mark size=0}
\addlegendentry{RACE} 

\addplot[color=black,mark=*,mark size=0] table {Data/OPT.RIDNC.M20.E10.P1.dat}; 
\addplot[densely dashed,color=black,mark=*,mark size=0] table {Data/RACE.M20.E10.P1.dat}; 

\end{axis}
\end{tikzpicture}

\caption{1 coded packets}
\label{fig:multi_RIDNC10_2}
\end{subfigure}
\begin{subfigure}[b]{1.65in}
\centering
\footnotesize
\begin{tikzpicture}
\begin{axis}[ xlabel=No. of users,
              ylabel=Gain $(\%)$,
              grid=major,
              log ticks with fixed point,
              xmin=10,
              xmax=1000,
              xtick={10,250,500,750,1000},
              ymin=0,
              ymax=110,
              ytick={0,20,...,100},
              height=1.55in,
              width=1.65in,
              legend pos=north east,
              y tick label style={/pgf/number format/fixed,
              /pgf/number format/1000 sep = \thinspace}
            ]

\addlegendimage{color=black,mark=*,mark size=0}
\addlegendentry{Opt. RIDNC} 
\addlegendimage{densely dashed,color=black,mark=*,mark size=0}
\addlegendentry{RACE} 

\addplot[color=black,mark=*,mark size=0] table {Data/OPT.RIDNC.M20.E10.P2.dat}; 
\addplot[densely dashed,color=black,mark=*,mark size=0] table {Data/RACE.M20.E10.P2.dat}; 

\end{axis}
\end{tikzpicture}

\caption{2 coded packets}
\label{fig:multi_RIDNC10_2}
\end{subfigure}
\begin{subfigure}[b]{1.65in}
\centering
\footnotesize
\begin{tikzpicture}
\begin{axis}[ xlabel=No. of users,
              ylabel=Gain $(\%)$,
              grid=major,
              log ticks with fixed point,
              xmin=10,
              xmax=1000,
              xtick={10,250,500,750,1000},
              ymin=0,
              ymax=110,
              ytick={0,20,...,100},
              height=1.55in,
              width=1.65in,
              legend pos=south east,
              y tick label style={/pgf/number format/fixed,
              /pgf/number format/1000 sep = \thinspace}
            ]

\addlegendimage{color=black,mark=*,mark size=0}
\addlegendentry{Opt. RIDNC} 
\addlegendimage{densely dashed,color=black,mark=*,mark size=0}
\addlegendentry{RACE}

\addplot[color=black,mark=*,mark size=0] table {Data/OPT.RIDNC.M20.E10.P3.dat}; 
\addplot[densely dashed,color=black,mark=*,mark size=0] table {Data/RACE.M20.E10.P3.dat}; 

\end{axis}
\end{tikzpicture}

\caption{3 coded packets}
\label{fig:multi_RIDNC10_4}
\end{subfigure}
\begin{subfigure}[b]{1.65in}
\centering
\footnotesize
\begin{tikzpicture}
\begin{axis}[ xlabel=No. of users,
              ylabel=Gain $(\%)$,
              grid=major,
              log ticks with fixed point,
              xmin=10,
              xmax=1000,
              xtick={10,250,500,750,1000},
              ymin=0,
              ymax=110,
              ytick={0,20,...,100},
              height=1.55in,
              width=1.65in,
              legend pos=south east,
              y tick label style={/pgf/number format/fixed,
              /pgf/number format/1000 sep = \thinspace}
            ]

\addlegendimage{color=black,mark=*,mark size=0}
\addlegendentry{Opt. RIDNC} 
\addlegendimage{densely dashed,color=black,mark=*,mark size=0}
\addlegendentry{RACE}

\addplot[color=black,mark=*,mark size=0] table {Data/OPT.RIDNC.M20.E10.P4.dat}; 
\addplot[densely dashed,color=black,mark=*,mark size=0] table {Data/RACE.M20.E10.P4.dat};  

\end{axis}
\end{tikzpicture}

\caption{4 coded packets}
\label{fig:multi_RIDNC10_8}
\end{subfigure}
\caption{
 Comparing the gain of RACE and the optimal RINDC.  
 The number of packets is $M=20$ and the erasure rates are uniformly distributed in the interval $(0,0.1)$. 
 }
\label{fig:multi_RIDNC10}
\end{figure*}


\begin{figure*}
\centering
\begin{subfigure}[b]{1.65in}
\centering
\footnotesize
\begin{tikzpicture}
\begin{axis}[ xlabel=No. of users,
              ylabel=Gain $(\%)$,
              grid=major,
              log ticks with fixed point,
              xmin=10,
              xmax=1000,
              xtick={10,250,500,750,1000},
              ymin=0,
              ymax=110,
              ytick={0,20,...,100},
              height=1.55in,
              width=1.65in,
              legend pos=north east,
              y tick label style={/pgf/number format/fixed,
              /pgf/number format/1000 sep = \thinspace}
            ]

\addlegendimage{color=black,mark=*,mark size=0}
\addlegendentry{Opt. RIDNC} 
\addlegendimage{densely dashed,color=black,mark=*,mark size=0}
\addlegendentry{RACE} 

\addplot[color=black,mark=*,mark size=0] table {Data/OPT.RIDNC.M20.E20.P1.dat}; 
\addplot[densely dashed,color=black,mark=*,mark size=0] table {Data/RACE.M20.E20.P1.dat}; 

\end{axis}
\end{tikzpicture}

\caption{1 coded packets}
\label{fig:multi_RIDNC10_2}
\end{subfigure}
\begin{subfigure}[b]{1.65in}
\centering
\footnotesize
\begin{tikzpicture}
\begin{axis}[ xlabel=No. of users,
              ylabel=Gain $(\%)$,
              grid=major,
              log ticks with fixed point,
              xmin=10,
              xmax=1000,
              xtick={10,250,500,750,1000},
              ymin=0,
              ymax=110,
              ytick={0,20,...,100},
              height=1.55in,
              width=1.65in,
              legend pos=north east,
              y tick label style={/pgf/number format/fixed,
              /pgf/number format/1000 sep = \thinspace}
            ]

\addlegendimage{color=black,mark=*,mark size=0}
\addlegendentry{Opt. RIDNC} 
\addlegendimage{densely dashed,color=black,mark=*,mark size=0}
\addlegendentry{RACE} 

\addplot[color=black,mark=*,mark size=0] table {Data/OPT.RIDNC.M20.E20.P2.dat}; 
\addplot[densely dashed,color=black,mark=*,mark size=0] table {Data/RACE.M20.E20.P2.dat}; 

\end{axis}
\end{tikzpicture}

\caption{2 coded packets}
\label{fig:multi_RIDNC10_2}
\end{subfigure}
\begin{subfigure}[b]{1.65in}
\centering
\footnotesize
\begin{tikzpicture}
\begin{axis}[ xlabel=No. of users,
              ylabel=Gain $(\%)$,
              grid=major,
              log ticks with fixed point,
              xmin=10,
              xmax=1000,
              xtick={10,250,500,750,1000},
              ymin=0,
              ymax=110,
              ytick={0,20,...,100},
              height=1.55in,
              width=1.65in,
              legend pos=south east,
              y tick label style={/pgf/number format/fixed,
              /pgf/number format/1000 sep = \thinspace}
            ]

\addlegendimage{color=black,mark=*,mark size=0}
\addlegendentry{Opt. RIDNC} 
\addlegendimage{densely dashed,color=black,mark=*,mark size=0}
\addlegendentry{RACE}

\addplot[color=black,mark=*,mark size=0] table {Data/OPT.RIDNC.M20.E20.P3.dat}; 
\addplot[densely dashed,color=black,mark=*,mark size=0] table {Data/RACE.M20.E20.P3.dat}; 

\end{axis}
\end{tikzpicture}

\caption{3 coded packets}
\label{fig:multi_RIDNC10_4}
\end{subfigure}
\begin{subfigure}[b]{1.65in}
\centering
\footnotesize
\begin{tikzpicture}
\begin{axis}[ xlabel=No. of users,
              ylabel=Gain $(\%)$,
              grid=major,
              log ticks with fixed point,
              xmin=10,
              xmax=1000,
              xtick={10,250,500,750,1000},
              ymin=0,
              ymax=110,
              ytick={0,20,...,100},
              height=1.55in,
              width=1.65in,
              legend pos=south east,
              y tick label style={/pgf/number format/fixed,
              /pgf/number format/1000 sep = \thinspace}
            ]

\addlegendimage{color=black,mark=*,mark size=0}
\addlegendentry{Opt. RIDNC} 
\addlegendimage{densely dashed,color=black,mark=*,mark size=0}
\addlegendentry{RACE}

\addplot[color=black,mark=*,mark size=0] table {Data/OPT.RIDNC.M20.E20.P4.dat}; 
\addplot[densely dashed,color=black,mark=*,mark size=0] table {Data/RACE.M20.E20.P4.dat}; 

\end{axis}
\end{tikzpicture}

\caption{4 coded packets}
\label{fig:multi_RIDNC10_8}
\end{subfigure}
\caption{
 Comparing the gain of RACE and the optimal RINDC.  
 The number of packets is $M=20$ and the erasure rates are uniformly distributed in the interval $(0,0.2)$. 
 }
\label{fig:multi_RIDNC20}
\end{figure*}
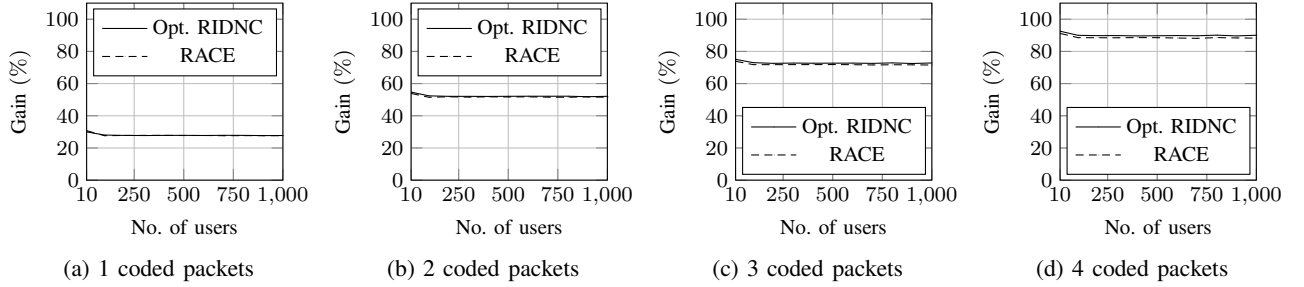


\begin{figure*}
\centering
\begin{subfigure}[b]{1.65in}
\centering
\footnotesize
\begin{tikzpicture}
\begin{axis}[ xlabel=No. of users,
              ylabel=Gain $(\%)$,
              grid=major,
              log ticks with fixed point,
              xmin=10,
              xmax=1000,
              xtick={10,250,500,750,1000},
              ymin=0,
              ymax=110,
              ytick={0,20,...,100},
              height=1.55in,
              width=1.65in,
              legend pos=north east,
              y tick label style={/pgf/number format/fixed,
              /pgf/number format/1000 sep = \thinspace}
            ]

\addlegendimage{color=black,mark=*,mark size=0}
\addlegendentry{Opt. RIDNC} 
\addlegendimage{densely dashed,color=black,mark=*,mark size=0}
\addlegendentry{RACE} 

\addplot[color=black,mark=*,mark size=0] table {Data/OPT.RIDNC.M20.E30.P1.dat}; 
\addplot[densely dashed,color=black,mark=*,mark size=0] table {Data/RACE.M20.E30.P1.dat}; 

\end{axis}
\end{tikzpicture}

\caption{1 coded packets}
\label{fig:multi_RIDNC10_2}
\end{subfigure}
\begin{subfigure}[b]{1.65in}
\centering
\footnotesize
\begin{tikzpicture}
\begin{axis}[ xlabel=No. of users,
              ylabel=Gain $(\%)$,
              grid=major,
              log ticks with fixed point,
              xmin=10,
              xmax=1000,
              xtick={10,250,500,750,1000},
              ymin=0,
              ymax=110,
              ytick={0,20,...,100},
              height=1.55in,
              width=1.65in,
              legend pos=north east,
              y tick label style={/pgf/number format/fixed,
              /pgf/number format/1000 sep = \thinspace}
            ]

\addlegendimage{color=black,mark=*,mark size=0}
\addlegendentry{Opt. RIDNC} 
\addlegendimage{densely dashed,color=black,mark=*,mark size=0}
\addlegendentry{RACE} 

\addplot[color=black,mark=*,mark size=0] table {Data/OPT.RIDNC.M20.E30.P2.dat}; 
\addplot[densely dashed,color=black,mark=*,mark size=0] table {Data/RACE.M20.E30.P2.dat}; 

\end{axis}
\end{tikzpicture}

\caption{2 coded packets}
\label{fig:multi_RIDNC10_2}
\end{subfigure}
\begin{subfigure}[b]{1.65in}
\centering
\footnotesize
\begin{tikzpicture}
\begin{axis}[ xlabel=No. of users,
              ylabel=Gain $(\%)$,
              grid=major,
              log ticks with fixed point,
              xmin=10,
              xmax=1000,
              xtick={10,250,500,750,1000},
              ymin=0,
              ymax=110,
              ytick={0,20,...,100},
              height=1.55in,
              width=1.65in,
              legend pos=south east,
              y tick label style={/pgf/number format/fixed,
              /pgf/number format/1000 sep = \thinspace}
            ]

\addlegendimage{color=black,mark=*,mark size=0}
\addlegendentry{Opt. RIDNC} 
\addlegendimage{densely dashed,color=black,mark=*,mark size=0}
\addlegendentry{RACE}

\addplot[color=black,mark=*,mark size=0] table {Data/OPT.RIDNC.M20.E30.P3.dat}; 
\addplot[densely dashed,color=black,mark=*,mark size=0] table {Data/RACE.M20.E30.P3.dat}; 

\end{axis}
\end{tikzpicture}

\caption{3 coded packets}
\label{fig:multi_RIDNC10_4}
\end{subfigure}
\begin{subfigure}[b]{1.65in}
\centering
\footnotesize
\begin{tikzpicture}
\begin{axis}[ xlabel=No. of users,
              ylabel=Gain $(\%)$,
              grid=major,
              log ticks with fixed point,
              xmin=10,
              xmax=1000,
              xtick={10,250,500,750,1000},
              ymin=0,
              ymax=110,
              ytick={0,20,...,100},
              height=1.55in,
              width=1.65in,
              legend pos=south east,
              y tick label style={/pgf/number format/fixed,
              /pgf/number format/1000 sep = \thinspace}
            ]

\addlegendimage{color=black,mark=*,mark size=0}
\addlegendentry{Opt. RIDNC} 
\addlegendimage{densely dashed,color=black,mark=*,mark size=0}
\addlegendentry{RACE}

\addplot[color=black,mark=*,mark size=0] table {Data/OPT.RIDNC.M20.E30.P4.dat}; 
\addplot[densely dashed,color=black,mark=*,mark size=0] table {Data/RACE.M20.E30.P4.dat}; 

\end{axis}
\end{tikzpicture}

\caption{4 coded packets}
\label{fig:multi_RIDNC10_8}
\end{subfigure}
\caption{
 Comparing the gain of RACE and the optimal RINDC.  
 The number of packets is $M=20$ and the erasure rates are uniformly distributed in the interval $(0,0.3)$. 
 }
\label{fig:multi_RIDNC30}
\end{figure*}
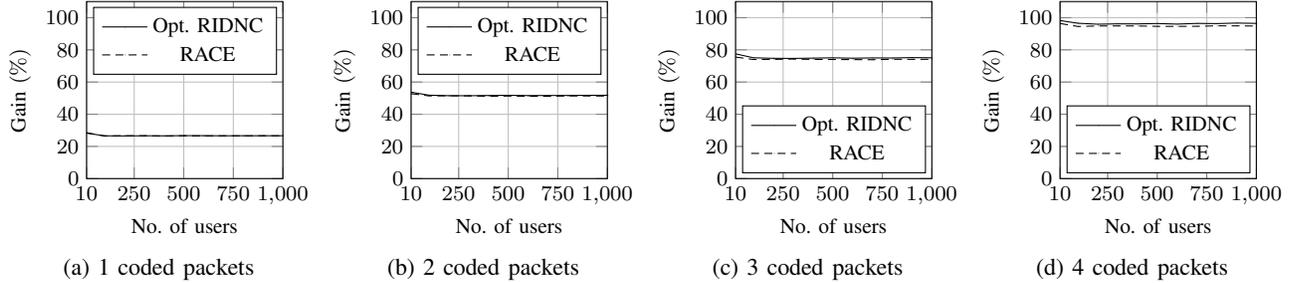


\begin{figure*}
\centering
\begin{subfigure}[b]{2.3in}
\centering
\footnotesize
\begin{tikzpicture}
\begin{axis}[ xlabel=No. of feedback,
              ylabel=Gain $(\%)$,
              grid=major,
              xmin=0,
              xmax=100,
              xtick={1,20,40,...,100},
              ymin=0,
              ymax=40,
              ytick={0,10,20,...,40},
              height=1.55in,
              width=2.3in,
              legend pos=south east,
              y tick label style={/pgf/number format/fixed,
              /pgf/number format/1000 sep = \thinspace}
            ] 

\addlegendimage{densely dashed,color=black,mark=*,mark size=0}
\addlegendentry{\scriptsize Feedback from all users} 
\addlegendimage{color=black,mark=*,mark size=0}
\addlegendentry{\scriptsize Varied no. of feedback}

\addplot[color=black,mark=*,mark size=0] table {Data/RACE.M20.N500.E10.P1.dat}; 
\addplot[densely dashed,color=black,mark=*,mark size=0,domain=0:100] {28.9760};

\end{axis}

\end{tikzpicture}

\caption{$\mathcal{E}\leq 10\%$.}
\label{fig:stat10}
\end{subfigure}
\begin{subfigure}[b]{2.3in}
\centering
\footnotesize
\begin{tikzpicture}
\begin{axis}[ xlabel=No. of feedback,
              ylabel=Gain $(\%)$,
              grid=major,
              xmin=0,
              xmax=100,
              xtick={1,20,40,...,100},
              ymin=0,
              ymax=40,
              ytick={0,10,20,...,40},
              height=1.55in,
              width=2.3in,
              legend pos=south east,
              y tick label style={/pgf/number format/fixed,
              /pgf/number format/1000 sep = \thinspace}
            ] 

\addlegendimage{densely dashed,color=black,mark=*,mark size=0}
\addlegendentry{\scriptsize Feedback from all users} 
\addlegendimage{color=black,mark=*,mark size=0}
\addlegendentry{\scriptsize Varied no. of feedback} 

\addplot[color=black,mark=*,mark size=0] table {Data/RACE.M20.N500.E20.P1.dat}; 
\addplot[densely dashed,color=black,mark=*,mark size=0,domain=0:100] {27.7386};
\end{axis}
\end{tikzpicture}

\caption{$\mathcal{E}\leq 20\%$.}
\label{fig:stat20}
\end{subfigure}
\begin{subfigure}[b]{2.3in}
\centering
\footnotesize
\begin{tikzpicture}
\begin{axis}[ xlabel=No. of feedback,
              ylabel=Gain $(\%)$,
              grid=major,
              xmin=0,
              xmax=100,
              xtick={1,20,40,...,100},
              ymin=0,
              ymax=40,
              ytick={0,10,20,...,40},
              height=1.55in,
              width=2.3in,
              legend pos=south east,
              y tick label style={/pgf/number format/fixed,
              /pgf/number format/1000 sep = \thinspace}
            ] 

\addlegendimage{densely dashed,color=black,mark=*,mark size=0}
\addlegendentry{\scriptsize Feedback from all users} 
\addlegendimage{color=black,mark=*,mark size=0}
\addlegendentry{\scriptsize Varied no. of feedback} 

\addplot[color=black,mark=*,mark size=0] table {Data/RACE.M20.N500.E30.P1.dat}; 
\addplot[densely dashed,color=black,mark=*,mark size=0,domain=0:100] {26.5884};
\end{axis}
\end{tikzpicture}

\caption{$\mathcal{E}\leq 30\%$.}
\label{fig:stat30}
\end{subfigure}
\caption{With up to 4, 8, and 13 feedback respectively, RACE achieves at least $90\%$ of its full gain.  
This result is independent of the number of users. Here, we set the number of users to 500, and the number of packets to $M=20$.}
\label{fig:stat}
\end{figure*}
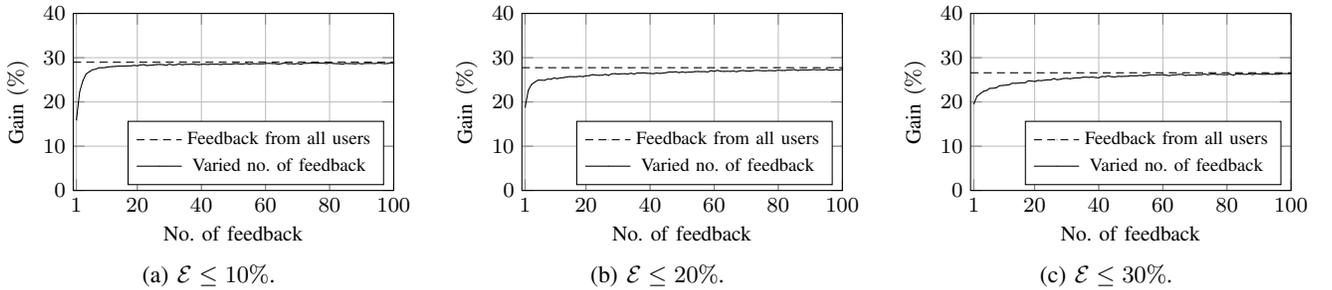


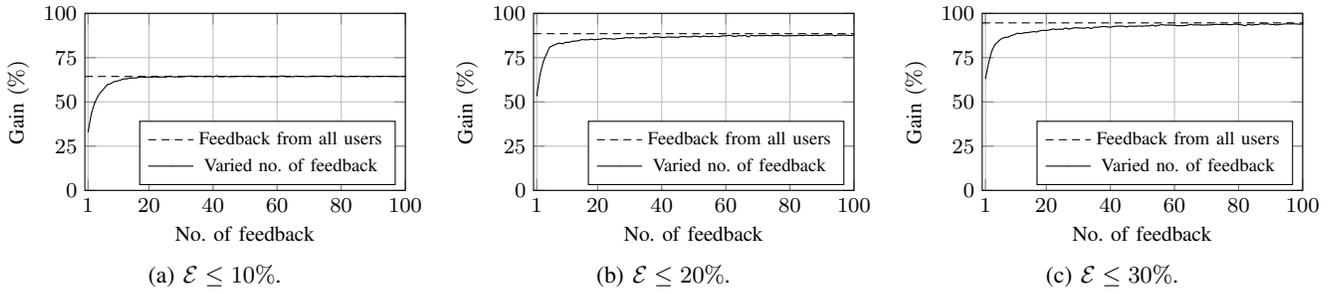
\begin{figure*}
\centering
\begin{subfigure}[b]{2.3in}
\centering
\footnotesize
\begin{tikzpicture}
\begin{axis}[ xlabel=No. of feedback,
              ylabel=Gain $(\%)$,
              grid=major,
              xmin=0,
              xmax=100,
              xtick={1,20,40,...,100},
              ymin=0,
              ymax=100,
              ytick={0,25,...,100},
              height=1.55in,
              width=2.3in,
              legend pos=south east,
              y tick label style={/pgf/number format/fixed,
              /pgf/number format/1000 sep = \thinspace}
            ] 

\addlegendimage{densely dashed,color=black,mark=*,mark size=0}
\addlegendentry{\scriptsize Feedback from all users} 
\addlegendimage{color=black,mark=*,mark size=0}
\addlegendentry{\scriptsize Varied no. of feedback} 

\addplot[color=black,mark=*,mark size=0] table {Data/RACE.M20.N500.E10.P4.dat}; 
\addplot[densely dashed,color=black,mark=*,mark size=0,domain=0:100] {64.3726};

\end{axis}

\end{tikzpicture}

\caption{$\mathcal{E}\leq 10\%$.}
\label{fig:multiStat10}
\end{subfigure}
\begin{subfigure}[b]{2.3in}
\centering
\footnotesize
\begin{tikzpicture}
\begin{axis}[ xlabel=No. of feedback,
              ylabel=Gain $(\%)$,
              grid=major,
              xmin=0,
              xmax=100,
              xtick={1,20,40,...,100},
              ymin=0,
              ymax=100,
              ytick={0,25,...,100},
              height=1.55in,
              width=2.3in,
              legend pos=south east,
              y tick label style={/pgf/number format/fixed,
              /pgf/number format/1000 sep = \thinspace}
            ] 

\addlegendimage{densely dashed,color=black,mark=*,mark size=0}
\addlegendentry{\scriptsize Feedback from all users} 
\addlegendimage{color=black,mark=*,mark size=0}
\addlegendentry{\scriptsize Varied no. of feedback}

\addplot[color=black,mark=*,mark size=0] table {Data/RACE.M20.N500.E20.P4.dat}; 
\addplot[densely dashed,color=black,mark=*,mark size=0,domain=0:100] {88.5300};
\end{axis}
\end{tikzpicture}

\caption{$\mathcal{E}\leq 20\%$.}
\label{fig:multiStat20}
\end{subfigure}
\begin{subfigure}[b]{2.3in}
\centering
\footnotesize
\begin{tikzpicture}
\begin{axis}[ xlabel=No. of feedback,
              ylabel=Gain $(\%)$,
              grid=major,
              xmin=0,
              xmax=100,
              xtick={1,20,40,...,100},
              ymin=0,
              ymax=100,
              ytick={0,25,...,100},
              height=1.55in,
              width=2.3in,
              legend pos=south east,
              y tick label style={/pgf/number format/fixed,
              /pgf/number format/1000 sep = \thinspace}
            ] 

\addlegendimage{densely dashed,color=black,mark=*,mark size=0}
\addlegendentry{\scriptsize Feedback from all users} 
\addlegendimage{color=black,mark=*,mark size=0}
\addlegendentry{\scriptsize Varied no. of feedback} 

\addplot[color=black,mark=*,mark size=0] table {Data/RACE.M20.N500.E30.P4.dat}; 
\addplot[densely dashed,color=black,mark=*,mark size=0,domain=0:100] {94.6216};
\end{axis}
\end{tikzpicture}

\caption{$\mathcal{E}\leq 30\%$.}
\label{fig:multiStat30}
\end{subfigure}
\caption{With up to 7, 5, and 6  feedback respectively, RACE achieves at least $90\%$ of its full gain.
  This result is independent of the number of users. 
  Here, we set $N=500$, $M=20$, and the number of coded packets to four.
}
\label{fig:multiStat}
\end{figure*}


\begin{figure*}
\centering
\begin{subfigure}[b]{2.3in}
\centering
\footnotesize
\begin{tikzpicture}
\begin{axis}[ xlabel=No. of feedback,
              ylabel=Gain $(\%)$,
              grid=major,
              xmin=0,
              xmax=50,
              xtick={1,10,20,...,50},
              ymin=0,
              ymax=40,
              height=1.55in,
              width=2.3in,
              legend pos=south east,
              y tick label style={/pgf/number format/fixed,
              /pgf/number format/1000 sep = \thinspace}
            ] 

\addlegendimage{densely dashed,color=black,mark=*,mark size=0}
\addlegendentry{RACE} 
\addlegendimage{color=black,mark=*,mark size=0}
\addlegendentry{CrowdWiFi} 

\addplot[color=black,mark=*,mark size=0] table {Data/CROWD.M20.N100.E20.P1.dat}; 
\addplot[densely dashed,color=black,mark=*,mark size=0] table {Data/RACE.M20.N100.E20.P1.dat}; 

\end{axis}

\end{tikzpicture}

\caption{100 Users.}
\label{fig:multiStat10}
\end{subfigure}
\begin{subfigure}[b]{2.3in}
\centering
\footnotesize
\begin{tikzpicture}
\begin{axis}[ xlabel=No. of feedback,
              ylabel=Gain $(\%)$,
              grid=major,
              xmin=0,
              xmax=100,
              xtick={1,20,40,...,100},
              ymin=0,
              ymax=40,
              height=1.55in,
              width=2.3in,
              legend pos=south east,
              y tick label style={/pgf/number format/fixed,
              /pgf/number format/1000 sep = \thinspace}
            ] 

\addlegendimage{densely dashed,color=black,mark=*,mark size=0}
\addlegendentry{RACE} 
\addlegendimage{color=black,mark=*,mark size=0}
\addlegendentry{CrowdWiFi} 

\addplot[color=black,mark=*,mark size=0] table {Data/CROWD.M20.N500.E20.P1.dat}; 
\addplot[densely dashed,color=black,mark=*,mark size=0] table {Data/RACE.M20.N500.E20.P1.dat}; 

\end{axis}
\end{tikzpicture}

\caption{500 Users.}
\label{fig:multiStat20}
\end{subfigure}
\begin{subfigure}[b]{2.3in}
\centering
\footnotesize
\begin{tikzpicture}
\begin{axis}[ xlabel=No. of feedback,
              ylabel=Gain $(\%)$,
              grid=major,
              xmin=0,
              xmax=200,
              xtick={1,40,80,...,200},
              ymin=0,
              ymax=40,
              height=1.55in,
              width=2.3in,
              legend pos=south east,
              y tick label style={/pgf/number format/fixed,
              /pgf/number format/1000 sep = \thinspace}
            ] 

\addlegendimage{densely dashed,color=black,mark=*,mark size=0}
\addlegendentry{RACE} 
\addlegendimage{color=black,mark=*,mark size=0}
\addlegendentry{CrowdWiFi} 

\addplot[color=black,mark=*,mark size=0] table {Data/CROWD.M20.N1000.E20.P1.dat}; 
\addplot[densely dashed,color=black,mark=*,mark size=0] table {Data/RACE.M20.N1000.E20.P1.dat}; 

\end{axis}
\end{tikzpicture}
\caption{1000 Users.}
\label{fig:multiStat30}
\end{subfigure}
\caption{The CrowdWiFi encoder requires feedback from about $15\%$ of all users to 
recover as many lost packets as RACE.
 The number of packets is $M=20$ and the erasure rates are uniformly distributed in the interval $(0,0.2)$. 
}
\label{fig:RvsCGainE20}
\end{figure*}
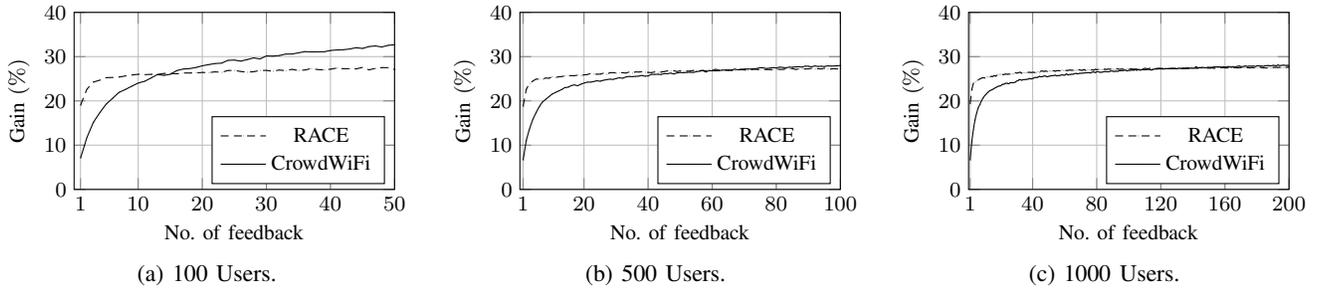

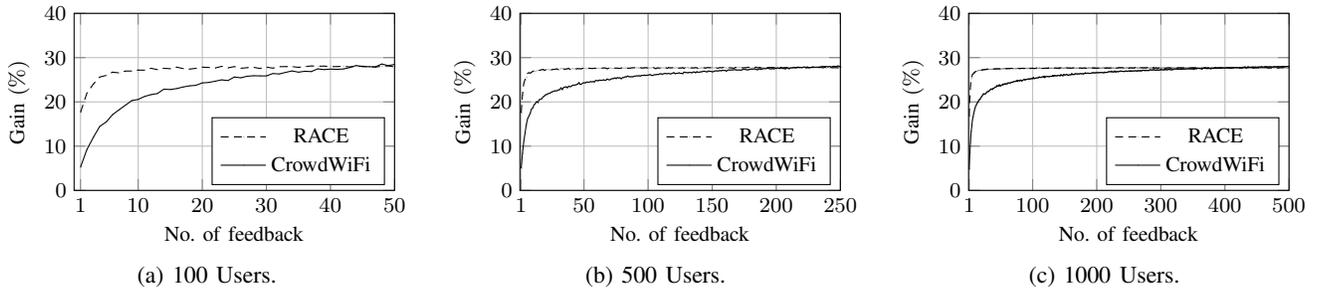
\begin{figure*}
\centering
\begin{subfigure}[b]{2.3in}
\centering
\footnotesize
\begin{tikzpicture}
\begin{axis}[ xlabel=No. of feedback,
              ylabel=Gain $(\%)$,
              grid=major,
              xmin=0,
              xmax=50,
              xtick={1,10,20,...,50},
              ymin=0,
              ymax=40,
              height=1.55in,
              width=2.3in,
              legend pos=south east,
              y tick label style={/pgf/number format/fixed,
              /pgf/number format/1000 sep = \thinspace}
            ] 

\addlegendimage{densely dashed,color=black,mark=*,mark size=0}
\addlegendentry{RACE} 
\addlegendimage{color=black,mark=*,mark size=0}
\addlegendentry{CrowdWiFi} 

\addplot[color=black,mark=*,mark size=0] table {Data/CROWD.M10.N100.E20.P1.dat}; 
\addplot[densely dashed,color=black,mark=*,mark size=0] table {Data/RACE.M10.N100.E20.P1.dat}; 

\end{axis}

\end{tikzpicture}

\caption{100 Users.}
\label{fig:multiStat10}
\end{subfigure}
\begin{subfigure}[b]{2.3in}
\centering
\footnotesize
\begin{tikzpicture}
\begin{axis}[ xlabel=No. of feedback,
              ylabel=Gain $(\%)$,
              grid=major,
              xmin=0,
              xmax=250,
              xtick={1,50,100,...,250},
              ymin=0,
              ymax=40,
              height=1.55in,
              width=2.3in,
              legend pos=south east,
              y tick label style={/pgf/number format/fixed,
              /pgf/number format/1000 sep = \thinspace}
            ] 

\addlegendimage{densely dashed,color=black,mark=*,mark size=0}
\addlegendentry{RACE} 
\addlegendimage{color=black,mark=*,mark size=0}
\addlegendentry{CrowdWiFi} 

\addplot[color=black,mark=*,mark size=0] table {Data/CROWD.M10.N500.E20.P1.dat}; 
\addplot[densely dashed,color=black,mark=*,mark size=0] table {Data/RACE.M10.N500.E20.P1.dat}; 

\end{axis}
\end{tikzpicture}

\caption{500 Users.}
\label{fig:multiStat20}
\end{subfigure}
\begin{subfigure}[b]{2.3in}
\centering
\footnotesize
\begin{tikzpicture}
\begin{axis}[ xlabel=No. of feedback,
              ylabel=Gain $(\%)$,
              grid=major,
              xmin=0,
              xmax=500,
              xtick={1,100,200,...,500},
              ymin=0,
              ymax=40,
              height=1.55in,
              width=2.3in,
              legend pos=south east,
              y tick label style={/pgf/number format/fixed,
              /pgf/number format/1000 sep = \thinspace}
            ] 

\addlegendimage{densely dashed,color=black,mark=*,mark size=0}
\addlegendentry{RACE} 
\addlegendimage{color=black,mark=*,mark size=0}
\addlegendentry{CrowdWiFi} 

\addplot[color=black,mark=*,mark size=0] table {Data/CROWD.M10.N1000.E20.P1.dat}; 
\addplot[densely dashed,color=black,mark=*,mark size=0] table {Data/RACE.M10.N1000.E20.P1.dat}; 

\end{axis}
\end{tikzpicture}
\caption{1000 Users.}
\label{fig:multiStat30}
\end{subfigure}
\caption{The CrowdWiFi encoder requires feedback from about $40\%$  of all users to 
recover as many lost packets as RACE.
 The number of packets is $M=10$ and the erasure rates are uniformly distributed in the interval $(0,0.2)$. 
}
\label{fig:RvsCGainE10}
\end{figure*}

\begin{figure*}
\centering
\begin{subfigure}[b]{2.3in}
\centering
\footnotesize
\begin{tikzpicture}
\begin{axis}[ xlabel=No. of feedback,
              ylabel=Speedup,
              grid=major,
              xmin=0,
              xmax=100,
              xtick={1,20,40,...,100},
              ymin=0,
              ymax=600,
              ytick={0,150,300,450,600},
              height=1.55in,
              width=2.3in,
              legend pos=north west,
              y tick label style={/pgf/number format/fixed,
              /pgf/number format/1000 sep = \thinspace}
            ] 

\addlegendimage{densely dashed,color=black,mark=*,mark size=0}
\addlegendentry{$M=10$} 
\addlegendimage{color=black,mark=*,mark size=0}
\addlegendentry{$M=20$}

\addplot[densely dashed,color=black,mark=*,mark size=0] table {Data/RvsC.TimeRatio.M10.N500.E10.dat};  
\addplot[color=black,mark=*,mark size=0] table {Data/RvsC.TimeRatio.M20.N500.E10.dat};  

\end{axis}

\end{tikzpicture}

\caption{$\mathcal{E}\leq 10\%$.}
\label{fig:multiStat10}
\end{subfigure}
\begin{subfigure}[b]{2.3in}
\centering
\footnotesize
\begin{tikzpicture}
\begin{axis}[ xlabel=No. of feedback,
              ylabel=Speedup,
              grid=major,
              xmin=0,
              xmax=100,
              xtick={1,20,40,...,100},
              ymin=0,
              ymax=600,
              ytick={0,150,300,450,600},
              height=1.55in,
              width=2.3in,
              legend pos=north west,
              y tick label style={/pgf/number format/fixed,
              /pgf/number format/1000 sep = \thinspace}
            ] 

\addlegendimage{densely dashed,color=black,mark=*,mark size=0}
\addlegendentry{$M=10$} 
\addlegendimage{color=black,mark=*,mark size=0}
\addlegendentry{$M=20$}

\addplot[densely dashed,color=black,mark=*,mark size=0] table {Data/RvsC.TimeRatio.M10.N500.E20.dat}; 
\addplot[color=black,mark=*,mark size=0] table {Data/RvsC.TimeRatio.M20.N500.E20.dat};  

\end{axis}

\end{tikzpicture}

\caption{$\mathcal{E}\leq 20\%$.}
\label{fig:multiStat10}
\end{subfigure}
\begin{subfigure}[b]{2.3in}
\centering
\footnotesize
\begin{tikzpicture}
\begin{axis}[ xlabel=No. of feedback,
              ylabel=Speedup,
              grid=major,
              xmin=0,
              xmax=100,
              xtick={1,20,40,...,100},
              ymin=0,
              ymax=600,
              ytick={0,150,300,450,600},
              height=1.55in,
              width=2.3in,
              legend pos=north west,
              y tick label style={/pgf/number format/fixed,
              /pgf/number format/1000 sep = \thinspace}
            ] 

\addlegendimage{densely dashed,color=black,mark=*,mark size=0}
\addlegendentry{$M=10$} 
\addlegendimage{color=black,mark=*,mark size=0}
\addlegendentry{$M=20$}

\addplot[densely dashed,color=black,mark=*,mark size=0] table {Data/RvsC.TimeRatio.M10.N500.E30.dat};  
\addplot[color=black,mark=*,mark size=0] table {Data/RvsC.TimeRatio.M20.N500.E30.dat};  

\end{axis}

\end{tikzpicture}

\caption{$\mathcal{E}\leq 30\%$.}
\label{fig:multiStat10c}
\end{subfigure}
\caption{Speedup: the ratio of runtime of the CrowdWiFi encoder over that of RACE.
}
\label{fig:RvsCTime}
\end{figure*}
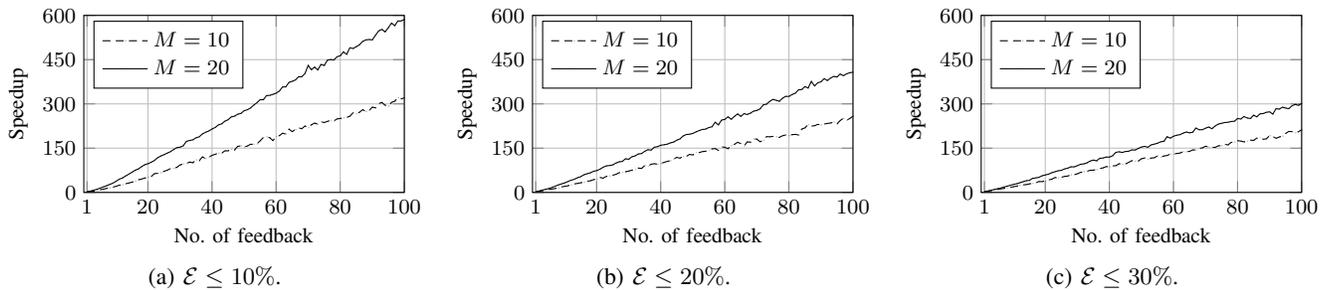



 \begin{appendices}

\clearpage
\section{Proof of Theorem~\ref{thm:stat} }
\label{app:stat}
    Let $\mathcal{C}$ denote the set of packets XORed to generate the coded packet, and 
    $x_{i,{\mathcal{C}}}$, $1\leq i \leq N$, $\mathcal{C}\subseteq \mathcal{P}$, be a binary random variable such that 
    \[
      x_{i,{\mathcal{C}}}=
      \begin{cases}
        1		& \text{if the coded packet 1) is received at user $i$, and}\\
        		& \text{2) results in a packet recovery at user $i$};\\
        0		& \text{Otherwise}\\
      \end{cases}
    \]
    Note that a user cannot recover any packet if it does not receive the coded packet.
    That is why the random variable  $x_{i,{\mathcal{C}}}$ is set to zero if the coded packet is not received at node $i$.
    Also, a received coded packet can result in a packet recovery  only if the user is missing exactly one packet from the set $\mathcal{C}$.
      
    The random variable $X_{\mathcal{C}}=\sum_{i=1}^{N}x_{i,{\mathcal{C}}} $ is the sum of $N$ independent binary
    random variables, thus, by a Chernoff bound, we get
    \begin{equation}
    \label{equ:ch-p1}
    \begin{split}
      Pr(X_{\mathcal{C}}>(1+\delta)\mu_{\mathcal{C}}) 
      \leq e^{\frac{-\delta^2\mu_{\mathcal{C}}}{3}}
    \end{split}
    \end{equation}   
    Let $P_{\mathcal{C}}=\frac{1}{N}\cdot \sum_{i=1}^{N} Pr(x_{i, \mathcal{C}}=1)$
    and
    $
      p^*=\max_{\mathcal{C}\subseteq \mathcal{P}} P_{\mathcal{C}}.
    $
    We have $\mu_{\mathcal{C}}=N\cdot P_{\mathcal{C}}$ and 
    $
      G_{R}=\max_{\mathcal{C}\in\mathcal{P}} \mu_{\mathcal{C}} = N\cdot p^*.
    $    
    By~(\ref{equ:ch-p1}), we get
    \[
    \begin{split}
      Pr\left(X_{\mathcal{C}}>(1+\delta)G_{R} \right) 
      &=Pr\left(X_{\mathcal{C}}>(1+\delta)\cdot \frac{G_{R}}{\mu_{\mathcal{C}}}\cdot \mu_{\mathcal{C}}\right) \\
      &=Pr\left(X_{\mathcal{C}}>(1+\delta)\cdot \frac{p^*}{P_{\mathcal{C}}} \cdot \mu_{\mathcal{C}}\right) \\      
      &\leq Pr\left(X_{\mathcal{C}}>(1+\underbrace{\delta\cdot \frac{p^*}{P_{\mathcal{C}}}}_{\delta'}) \cdot \mu_{\mathcal{C}}\right) \\            
      &\leq e^{\frac{-\delta'^2\cdot \mu_{\mathcal{C}}}{3}}\\
      &=  e^{\frac{-\delta'^2\cdot N\cdot P_{\mathcal{C}}}{3}}
      =  e^{\frac{-(\delta\cdot \frac{p^*}{P_{\mathcal{C}}})^2NP_{\mathcal{C}}}{3}} \\   
      &\leq  e^{-\delta^2Np^*}       
      \leq \frac{\epsilon}{2^M}.
    \end{split}
    \]    
    The total number of subsets $\mathcal{C}$ of $\mathcal{P}$ is $2^M$, and 
    $Pr(X_{\mathcal{C}}>(1+\delta)G_{R}) \leq \frac{\epsilon}{2^M}$ for every random variable $X_{\mathcal{C}}$.
    Therefore, by the union bound, the probability that $X_{\mathcal{C}}>(1+\delta)G_{R}$ for at least one set  $\mathcal{C}$ is at most 
    \[
    2^M\cdot \frac{\epsilon}{2^M}=\epsilon.
    \]
    Thus,
    \[
    \begin{split}
    G_{opt}
    =\max_{\mathcal{C}\subseteq \mathcal{P}} X_{\mathcal{C}} 
    \leq (1+\delta) \max_{\mathcal{C}\subseteq \mathcal{P}} \mu_{\mathcal{C}} 
    = (1+\delta)G_{R}.
    \end{split}
    \]  
    with probability at least $1-\epsilon$.   
    Therefore,
    \[
      G_{R} \geq (1-\delta)\cdot G_{opt},
    \]
    with probability at least $1-\epsilon$.

\section{Proof of Theorem~\ref{thm:expected} }
    \label{app:expected}
      Let the binary random variable $y'_{d,i}$ be the number of recovered packet at user $u_i$, and
      \[
        \Delta_i=\frac{{h_i \choose d-1}(M-h_i) } {{M \choose d}}.
      \]
      We have
      
      \[
      \begin{split}
        &E[y'_{d,i}|H_i]\\
        &= \int_{0}^{1} E[y'_{d,i}|H_i, \mathcal{E}_i=\epsilon] Pr(\mathcal{E}_i=\epsilon| H_i) d\epsilon \\
        &= \int_{0}^{1} \left( \frac{{h_i \choose d-1}(M-h_i) } {{M \choose d}}\cdot (1-\epsilon)\right)Pr(\mathcal{E}_i=\epsilon| H_i) d\epsilon \\    
        &=  \Delta_i\cdot \int_{0}^{1} (1-\epsilon)Pr(\mathcal{E}_i=\epsilon| H_i) d\epsilon \\
        &=  \Delta_i\cdot \int_{0}^{1} (1-\epsilon)\frac{Pr(\mathcal{E}_i=\epsilon)}{Pr(H_i)}\cdot Pr( H_i|\mathcal{E}_i=\epsilon) d\epsilon \\    
        &=  \Delta_i\cdot \int_{0}^{1} (1-\epsilon)\frac{Pr( H_i|\mathcal{E}_i=\epsilon) }{Pr(H_i)} d\epsilon \\     
        &=  \Delta_i\cdot \int_{0}^{1} (1-\epsilon)\frac{Pr( H_i|\mathcal{E}_i=\epsilon)}{\int_0^1 Pr(H_i|\mathcal{E}_i=x)Pr(\mathcal{E}_i=x) dx}   d\epsilon \\   
       &=  \Delta_i\cdot \frac{ \int_{0}^{1} (1-\epsilon)Pr( H_i|\mathcal{E}_i=\epsilon) d\epsilon}{\int_0^1 Pr(H_i|\mathcal{E}_i=x) dx}    \\ 
       &=  \Delta_i\cdot 
       \frac{ \int_{0}^{1} (1-\epsilon){M \choose d} (1-\epsilon)^{h_i}\epsilon^{M-h_i}  d\epsilon}
       {\int_0^1 {M \choose d} (1-x)^{h_i}x^{M-h_i}  dx}    \\            
       &=  \Delta_i\cdot 
       \frac{ \int_{0}^{1}  (1-\epsilon)^{h_i+1}\epsilon^{M-h_i}  d\epsilon}
       {\int_0^1 (1-x)^{h_i}x^{M-h_i}  dx}    \\  
       &=  \Delta_i\cdot 
       \frac{ \int_{0}^{1}  (1-\epsilon)^{h_i+1}\epsilon^{M+1-(h_i+1)}  d\epsilon}
       {\int_0^1 (1-x)^{h_i}x^{M-h_i}  dx}.    \\        
      \end{split}
      \]   
    We have
    \[
    \begin{split}
        &\forall \quad a,b,a-b \in \mathbb{Z}^{\ge 0} \\
        &f(a,b)= \int_0^1 (1-x)^{b}x^{a-b}  dx=\frac{b!(a-b)!}{(a+1)!},
    \end{split}    
    \]
    thus
    \[
    \begin{split}
        E[y'_{d,i}|H_i]
       &=  \Delta_i\cdot  \frac{f(M+1, h_i+1)}{f(M, h_i)}    \\
       &=\Delta_i\cdot  \frac{\frac{(b+1)!(a-b)!}{(a+2)!}}{\frac{b!(a-b)!}{(a+1)!}}\\
       &=\Delta_i\cdot  \frac{h_i+1}{M+2}.\\
    \end{split}
    \]
      Note that
      \[
        E[y'_{d,i}|H_1, H_2, \ldots, H_N]= E[y'_{d,i}|H_i],
      \]  
      because $y'_{d,i}$  is independent of the events $H_j$, $j\neq i$.
      We have
      \[
      \begin{split}
        E[Y'_d|H_1, H_2, \ldots, H_N] 
        &= \sum_{i=1}^{N}E[y'_{d,i}|H_1, H_2, \ldots, H_N] \\
        &= \sum_{i=1}^{N}E[y'_{d,i}|H_i] \\
        &= \sum_{i=1}^{N} \left( \frac{{h_i \choose d-1}(M-h_i) } {{M \choose d}}\cdot \frac{h_i+1}{M+2} \right)    
      \end{split}
      \]

\end{appendices}


\bibliographystyle{IEEEtran}
\bibliography{IEEEabrv,ref.bib}

\begin{thebibliography}{10}
\providecommand{\url}[1]{#1}
\csname url@samestyle\endcsname
\providecommand{\newblock}{\relax}
\providecommand{\bibinfo}[2]{#2}
\providecommand{\BIBentrySTDinterwordspacing}{\spaceskip=0pt\relax}
\providecommand{\BIBentryALTinterwordstretchfactor}{4}
\providecommand{\BIBentryALTinterwordspacing}{\spaceskip=\fontdimen2\font plus
\BIBentryALTinterwordstretchfactor\fontdimen3\font minus
  \fontdimen4\font\relax}
\providecommand{\BIBforeignlanguage}[2]{{%
\expandafter\ifx\csname l@#1\endcsname\relax
\typeout{** WARNING: IEEEtran.bst: No hyphenation pattern has been}%
\typeout{** loaded for the language `#1'. Using the pattern for}%
\typeout{** the default language instead.}%
\else
\language=\csname l@#1\endcsname
\fi
#2}}
\providecommand{\BIBdecl}{\relax}
\BIBdecl

\bibitem{cisco2017paper}
\BIBentryALTinterwordspacing
Cisco visual networking index: Forecast and methodology, 2016-2021. [Online].
  Available:
  \url{https://www.cisco.com/c/en/us/solutions/collateral/service-provider/visual-networking-index-vni/complete-white-paper-c11-481360.html}
\BIBentrySTDinterwordspacing

\bibitem{ferreira2014real}
D.~Ferreira, R.~A. Costa, and J.~Barros, ``Real-time network coding for live
  streaming in hyper-dense {WiFi} spaces,'' \emph{IEEE Journal on Selected
  Areas in Communications}, vol.~32, no.~4, pp. 773--781, 2014.

\bibitem{le2017recovery}
A.~Le, A.~S. Tehrani, A.~Dimakis, and A.~Markopoulou, ``Recovery of packet
  losses in wireless broadcast for real-time applications,'' \emph{IEEE/ACM
  Transactions on Networking}, vol.~25, no.~2, pp. 676--689, 2017.

\bibitem{keller2008online}
L.~Keller, E.~Drinea, and C.~Fragouli, ``Online broadcasting with network
  coding,'' in \emph{Workshop on Network Coding and Applications (NetCod)},
  2008.

\bibitem{bejerano2016amuse}
Y.~Bejerano, V.~Gupta, C.~Gutterman, and G.~Zussman, ``{AMuSe}: Adaptive
  multicast services to very large groups-project overview,'' in \emph{Computer
  Communication and Networks (ICCCN)}, 2016, pp. 1--9.

\bibitem{SalyersSP08}
D.~Salyers, A.~Striegel, and C.~Poellabauer, ``Wireless reliability: Rethinking
  802.11 packet loss,'' in \emph{International Symposium on a World of
  Wireless, Mobile and Multimedia Networks, {WOWMOM}}, 2008, pp. 1--4.

\bibitem{ahlswede2000network}
R.~Ahlswede, N.~Cai, S.-Y. Li, and R.~W. Yeung, ``Network information flow,''
  \emph{IEEE Transactions on information theory}, vol.~46, no.~4, pp.
  1204--1216, 2000.

\bibitem{eryilmaz2008delay}
A.~Eryilmaz, A.~Ozdaglar, M.~M{\'e}dard, and E.~Ahmed, ``On the delay and
  throughput gains of coding in unreliable networks,'' \emph{IEEE Transactions
  on Information Theory}, vol.~54, no.~12, pp. 5511--5524, 2008.

\bibitem{yu2014instantly}
M.~Yu, N.~Aboutorab, and P.~Sadeghi, ``From instantly decodable to random
  linear network coded broadcast,'' \emph{IEEE Transactions on Communications},
  vol.~62, no.~11, pp. 3943--3955, 2014.

\bibitem{ho2006random}
T.~Ho, M.~M{\'e}dard, R.~Koetter, D.~R. Karger, M.~Effros, J.~Shi, and
  B.~Leong, ``A random linear network coding approach to multicast,''
  \emph{IEEE Transactions on Information Theory}, vol.~52, no.~10, pp.
  4413--4430, 2006.

\bibitem{lucani2009broadcasting}
D.~E. Lucani, M.~M{\'e}dard, and M.~Stojanovic, ``Broadcasting in time-division
  duplexing: A random linear network coding approach,'' in \emph{Workshop on
  Network Coding and Applications (NetCod)}, 2009, pp. 62--67.

\bibitem{lima2007random}
L.~Lima, M.~M{\'e}dard, and J.~Barros, ``Random linear network coding: A free
  cipher?'' in \emph{International Symposium on Information Theory (ISIT)},
  2007, pp. 546--550.

\bibitem{li2003linear}
S.-Y. Li, R.~W. Yeung, and N.~Cai, ``Linear network coding,'' \emph{IEEE
  Transactions on information theory}, vol.~49, no.~2, pp. 371--381, 2003.

\bibitem{li2005theory}
S.-Y. Li, N.~Cai, and R.~W. Yeung, ``On theory of linear network coding,'' in
  \emph{International Symposium on Information Theory (ISIT)}, 2005, pp.
  273--277.

\bibitem{lucani2009random}
D.~E. Lucani, M.~Stojanovic, and M.~M{\'e}dard, ``Random linear network coding
  for time division duplexing: When to stop talking and start listening,'' in
  \emph{INFOCOM}, 2009, pp. 1800--1808.

\bibitem{shokrollahi2006raptor}
A.~Shokrollahi, ``Raptor codes,'' \emph{IEEE Transactions on information
  theory}, vol.~52, no.~6, pp. 2551--2567, 2006.

\bibitem{nguyen2009network}
D.~Nguyen and T.~Nguyen, ``Network coding-based wireless media transmission
  using pomdp,'' in \emph{Packet Video}, 2009, pp. 1--9.

\bibitem{costa2008informed}
R.~A. Costa, D.~Munaretto, J.~Widmer, and J.~Barros, ``Informed network coding
  for minimum decoding delay,'' in \emph{International Conference on Mobile Ad
  Hoc and Sensor Systems (MASS)}, 2008, pp. 80--91.

\bibitem{rayanchu2008loss}
S.~Rayanchu, S.~Sen, J.~Wu, S.~Banerjee, and S.~Sengupta, ``Loss-aware network
  coding for unicast wireless sessions: design, implementation, and performance
  evaluation,'' in \emph{ACM SIGMETRICS Performance Evaluation Review},
  vol.~36, no.~1.\hskip 1em plus 0.5em minus 0.4em\relax ACM, 2008, pp. 85--96.

\bibitem{seferoglu2007opportunistic}
H.~Seferoglu and A.~Markopoulou, ``Opportunistic network coding for video
  streaming over wireless,'' in \emph{Packet Video}, 2007, pp. 191--200.

\bibitem{chen2007opportunistic}
W.~Chen, K.~B. Letaief, and Z.~Cao, ``Opportunistic network coding for wireless
  networks,'' in \emph{International Conference on Communications (ICC)}, 2007,
  pp. 4634--4639.

\bibitem{seferoglu2009video}
H.~Seferoglu and A.~Markopoulou, ``Video-aware opportunistic network coding
  over wireless networks,'' \emph{IEEE Journal on Selected Areas in
  Communications}, vol.~27, no.~5, 2009.

\bibitem{fragouli2007feedback}
C.~Fragouli, D.~Lun, M.~M{\'e}dard, and P.~Pakzad, ``On feedback for network
  coding,'' in \emph{Annual Conference on Information Sciences and Systems
  (CISS)}, 2007, pp. 248--252.

\bibitem{le2013instantly}
A.~Le, A.~S. Tehrani, A.~G. Dimakis, and A.~Markopoulou, ``Instantly decodable
  network codes for real-time applications,'' in \emph{Workshop on Network
  Coding and Applications (NetCod)}, 2013, pp. 1--6.

\bibitem{muhammad2013instantly}
M.~Muhammad, M.~Berioli, G.~Liva, and G.~Giambene, ``Instantly decodable
  network coding protocols with unequal error protection,'' in
  \emph{International Conference on Communications (ICC)}, 2013, pp.
  5120--5125.

\bibitem{katti2008xors}
S.~Katti, H.~Rahul, W.~Hu, D.~Katabi, M.~M{\'e}dard, and J.~Crowcroft, ``{XORs}
  in the air: practical wireless network coding,'' \emph{IEEE/ACM Transactions
  on Networking (ToN)}, vol.~16, no.~3, pp. 497--510, 2008.

\bibitem{arefi2018blind}
A.~Arefi, M.~Khabbazian, M.~Ardakani, and G.~Bansal, ``Blind instantly
  decodable network codes for wireless broadcast of real-time multimedia,''
  \emph{IEEE Transactions on Wireless Communications}, vol.~17, no.~4, pp.
  2276--2288, 2018.

\bibitem{birk2006coding}
Y.~Birk and T.~Kol, ``Coding on demand by an informed source ({ISCOD}) for
  efficient broadcast of different supplemental data to caching clients,''
  \emph{IEEE Transactions on Information Theory}, vol.~52, no.~6, pp.
  2825--2830, 2006.

\bibitem{el2010coding}
S.~El~Rouayheb, A.~Sprintson, and P.~Sadeghi, ``On coding for cooperative data
  exchange,'' in \emph{Information Theory Workshop on Information Theory
  (ITW)}, 2010, pp. 1--5.

\bibitem{CrowdWiFi}
\emph{CrowdWiFi Streaming}, Streambolico, Porto, Portugal, 2016.

\end{thebibliography}

\end{document}